%% file: one-way-combined.tex
\documentclass[11pt]{article}
\usepackage{amssymb}
\usepackage{amsfonts}
\usepackage{amsmath}
\usepackage{amsthm}
\usepackage{latexsym}

\usepackage[pagebackref=true]{hyperref}
\hypersetup{
    unicode=false,          
    colorlinks=true,        
    linkcolor=red,          
    citecolor=blue,        
    filecolor=magenta,      
    urlcolor=cyan           
}

\usepackage{hyperref,cleveref}
\usepackage{amsmath}
\usepackage{amsthm}
\usepackage{amsfonts}
\usepackage{fullpage,appendix}
\usepackage{algorithmic}
\usepackage[ruled,vlined]{algorithm2e}
\usepackage{dsfont}
\usepackage{color}
\usepackage{tikz}

\newcommand{\ignore}[1]{}
\definecolor{corlinks}{RGB}{64,128,128}
\definecolor{cormenu}{RGB}{0,37,94}
\definecolor{corurl}{RGB}{0,46,91}

\newcommand{\1}{\mathds{1}}

\newcommand{\err}{\mathsf{err}}

\newcommand{\D}{\mathcal{D}}

\newcommand{\F}{\mathbb{F}}
\newcommand{\calf}{\mathcal{F}}

\newcommand{\zo}{\{0, 1\}}

\newcommand{\Ex}{\mathbb E}

\newcommand{\cald}{\mathcal{D}}

\newcommand{\Disc}{\textsc{Disc}}

\newenvironment{proofof}[1]{\noindent{\bf Proof}
of #1:\hspace*{1em}}{\qed\bigskip}

\newtheorem{fact}{Fact}[section]
\newtheorem{definition}[fact]{Definition}

\newtheorem{theorem}[fact]{Theorem}
\newtheorem{lemma}[fact]{Lemma}
\newtheorem{corollary}[fact]{Corollary}

\newtheorem{proposition}[fact]{Proposition}

\newtheorem{claim}[fact]{Claim}

\newcommand{\CC}{\mathsf{CC}} 
\newcommand{\CCU}{\mathsf{CCU}} 
\newcommand{\owCC}{\mathsf{owCC}} 
\newcommand{\owCCU}{\mathsf{owCCU}} 
\newcommand{\owF}{\mathsf{ow}\calf} 
\newcommand{\agree}{\textsc{Agree}}

\newcommand{\Inote}[1]{}
\newcommand{\Pnote}[1]{}
\newcommand{\Mnote}[1]{}

\title{Communication with Contextual Uncertainty} \author{
Badih Ghazi \thanks{Computer Science and Artificial Intelligence Laboratory, Massachusetts Institute of Technology, Cambridge MA 02139. Supported in part by NSF STC Award CCF 0939370 and NSF Award CCF-1217423.  {\tt badih@mit.edu}.}
\and
 Ilan Komargodski
  \thanks{Weizmann Institute of Science, Israel. Email:
  {\tt
      ilan.komargodski@weizmann.ac.il}. Work done while an intern at MSR New
    England.
    Supported in part by a grant from the I-CORE Program of the Planning and
    Budgeting Committee, the Israel Science Foundation, BSF and the Israeli
    Ministry of Science and Technology.}
\and Pravesh Kothari \thanks{
    UT Austin, USA. Email: {\tt kothari@cs.utexas.edu}. Work done while an
    intern at MSR New England.} 
    \and Madhu Sudan \thanks{Microsoft Research, One
    Memorial Drive, Cambridge, MA 02142, USA. Email: {\tt madhu@mit.edu}.}}
\date{\today}

\begin{document}
\maketitle

\input{abstract}

\thispagestyle{empty}

\newpage

\pagenumbering{arabic}

\input{intro}

\input{model}

\input{hardag}

\input{one-way-ub}

\input{two-way-lb}

\input{conc}

\bibliographystyle{alpha}
\bibliography{one-way-combined}

\end{document}

%% file: abstract.tex
\begin{abstract}
We introduce a simple model illustrating the role of context in communication
and the challenge posed by uncertainty of knowledge of context. We consider a
variant of distributional communication complexity where Alice gets some
information $x$ and Bob gets $y$, where $(x,y)$ is drawn from a known
distribution, and Bob wishes to compute some function $g(x,y)$ (with
high probability over $(x,y)$). In our variant, Alice does not know $g$, but only knows some
function $f$ which is an approximation of $g$. Thus, the function being computed
forms the context for the communication, and knowing it imperfectly models
(mild) uncertainty in this context.

A naive solution would be for Alice and Bob to first agree on some common
function $h$ that is close to both $f$ and $g$ and then use a protocol for $h$
to compute $h(x,y)$. We show that any such agreement leads to a large overhead
in communication ruling out such a universal solution.

In contrast, we show
that if $g$ has a one-way communication protocol with complexity $k$ in the
standard setting, then it has a communication protocol with complexity $O(k \cdot (1+I))$ in the uncertain setting, where $I$ denotes the mutual information between $x$ and $y$. In the particular case where the input distribution is a product distribution, the protocol in the uncertain setting only incurs a constant factor blow-up in communication and error.

Furthermore, we show that the dependence on the mutual information $I$ is required. Namely, we construct a class of functions along with a non-product distribution over $(x,y)$ for which the communication complexity is a single bit in the standard setting but at least $\Omega(\sqrt{n})$ bits in the uncertain setting.
\end{abstract}

%% file: intro.tex
\section{Introduction}

Most forms of communication involve communicating players that share a large common context
and use this context to compress communication. In natural settings, the context may include understanding
of language, and knowledge of the environment and laws. In designed (computer-to-computer) settings, the context includes knowledge of the operating system, communication protocols, and encoding/decoding mechanisms. Remarkably, especially in the natural setting, context can seemingly be used to compress communication, even when it is not shared perfectly.
This ability to communicate despite a major source of uncertainty has led to a series of works attempting to model
various forms of communication amid uncertainty,
starting with Goldreich, Juba and Sudan~\cite{JS1,GJS:jacm} followed by
\cite{JKKS11,JS2,JubaWilliams13,HaramatyS,CGMS15}.
This current work introduces a new theme to this series of works by introducing a functional notion
of uncertainty and studying this model. We start by describing our model and results below and then contrast our model with some of the previous works.

\paragraph{Model.}
Our model builds upon the classical setup of communication complexity due to Yao~\cite{Yao}, and we develop it here.
The classical model considers two interacting players Alice and Bob each possessing some private information $x$
and $y$ with $x$ known only to Alice and $y$ to Bob. They wish to compute some joint
function $g(x,y)$ and would like to do so while exchanging the minimum possible number of
bits. In this work, we suggest that the function $g$ is the {\em context} of the communication
and consider a setting where it is shared imperfectly. Specifically, we say that
Bob knows the function $g$ and Alice knows some approximation $f$ to $g$ (with $f$ not being known to Bob). This leads to the question: when can
Alice and Bob interact to compute $g(x,y)$ with limited communication ?

It is clear that if $x\in \{0,1\}^n$, then $n$ bits of communication suffice --- Alice can
simply ignore $f$ and send $x$ to Bob. We wish to consider settings that improve on this.
To do so correctly on every input, a necessary condition is that $g$ must have low communication complexity in the standard model. However, this necessary condition does not seem to be sufficient --- since Alice only has an approximation $f$ to $g$. Thus, we settle for a weaker goal: determining $g$ correctly only on most inputs. This puts us in a distributional communication complexity setting. A necessary condition now is that $g$ must have a low-error low-communication protocol in the standard setting. The question is then: can $g$ be computed with
low error and low communication when Alice only knows an approximation $f$ to $g$ (with $f$ being unknown to Bob) ?

More precisely, in this setting, the input to Alice is a pair $(f,x)$ and the input to Bob is a pair $(g,y)$.
The functions $(f,g)$ are adversarially chosen subject to the restrictions that they are close
to each other (under some distribution $\mu$ on the inputs) and that $g$ (and hence $f$) has a low-error
low-communication protocol. The pair $(x,y)$ is drawn from the distribution $\mu$ (independent
of the choice of $f$ and $g$). The players both know $\mu$ in addition to their respective inputs.

\paragraph{Results.}
In order to describe our results, we first introduce some notation. Let $\delta_\mu(f,g)$ denote the (weighted and normalized)
Hamming distance between $f$ and $g$ with respect to the distribution $\mu$. Let $\CC^\mu_\epsilon(f)$ denote the minimum
communication complexity of a protocol computing $f$ correctly on all but an $\epsilon$ fraction of the inputs.
Let $\owCC^\mu_\epsilon(f)$ denote the corresponding  {\em one-way} communication complexity of $f$.
Given a family $\calf$ of pairs of functions $(f,g)$, we denote the uncertain complexity
$\CCU^\mu_\epsilon(\calf)$ to be the minimum over all public-coin protocols $\Pi$ of the
maximum over $(f,g) \in \calf$, $(x,y)$ in the support of $\mu$ and settings of public coins, of the communication cost of $\Pi$, subject to the condition that for every $(f,g) \in \calf$, $\Pi$ outputs $g(x,y)$ with probability $1-\epsilon$ over the choice of $(x,y)$ and the shared randomness. That is,
\begin{align*}
  \CCU^\mu_\epsilon(\calf) \triangleq \min_{\{\Pi \; | \; \forall (f,g)\in\calf:\; \delta_\mu(\Pi,g)\leq
    \epsilon\}}\max_{\{(f,g) \in \calf, (x,y)\in\mathsf{supp}(\mu),\text{ public coins}\}}\{
  \text{Comm.\ cost of }\Pi((f,x),(g,y)) \}.
\end{align*}
Similarly, let $\owCCU^\mu_\epsilon(\calf)$  denote the one-way uncertain communication complexity of
$\calf$.

Our first result (Theorem~\ref{thm:one-way}) shows that if $\mu$ is a distribution on which $f$ and $g$ are close and each has a one-way protocol with communication $k$ bits in the standard model, then the pair $(f,g)$ has one-way uncertain communication complexity of at most $O(k\cdot(1+I))$ bits with $I$ being the mutual information\footnote{Given a distribution $\mu$ over a pair $(X,Y)$ of random variables with marginals $\mu_X$ and $\mu_Y$ over $X$ and $Y$ respectively, the \emph{mutual information} of $X$ and $Y$ is defined as $I(X;Y) \triangleq \Ex_{(x,y) \sim \mu}[\log(\frac{\mu(x,y)}{\mu_X(x) \mu_Y(y)})]$.} of $(x,y) \sim \mu$. More precisely, let $\owF_{k,\epsilon,\delta}$ denote the family of all pairs of functions
$(f,g)$ with $\owCC^\mu_\epsilon(f), \owCC^\mu_\epsilon(g) \leq k$ and
$\delta_\mu(f,g) \leq \delta$. We prove the following theorem.

\begin{theorem}\label{thm:one-way}
There exists an absolute constant $c$ such that for every pair of finite sets $X$ and $Y$, every distribution $\mu$ over $X \times Y$ and every $\theta > 0$, it holds that
  \begin{equation}\label{eq:main_up_bd_intro}
    \owCCU_{\epsilon + 2\delta + \theta}^{\mu}(\owF^\mu_{k,\epsilon,\delta})
    \leq \frac{c \left(k+\log\left(\frac{1}{\theta}\right)\right)}{\theta^2} \cdot \left(1+\frac{I(X;Y)}{\theta^2}\right).
  \end{equation}
\end{theorem}

In the special case where $\mu$ is a product distribution, then $I(X;Y) = 0$ and we obtain the following particularly interesting corollary of \Cref{thm:one-way}.

\begin{corollary}\label{cor:one-way_product}
  There exists an absolute constant $c$ such that for every pair of finite sets $X$ and $Y$, every \emph{product} distribution $\mu$ over $X \times Y$ and every $\theta > 0$, it holds that
  \begin{equation*}\label{eq:main_up_bd_cor}
    \owCCU_{\epsilon + 2\delta + \theta}^{\mu}(\owF^\mu_{k,\epsilon,\delta})
    \leq \frac{c \left(k+\log\left(\frac{1}{\theta}\right)\right)}{\theta^2}.
  \end{equation*}
\end{corollary}

In words, Corollary~\ref{cor:one-way_product} says that for product distributions and for constant error probabilities, communication in the uncertain model is only a constant factor larger than in the standard model.

Our result is significant in that it achieves (moderately) reliable communication despite uncertainty
about the context, even when the uncertainty itself is hard to resolve. To elaborate on this statement,
note that one hope for
achieving a low-communication protocol for $g$ would be for Alice and Bob to first agree on some function
$q$ that is close to $f$ and $g$, and then apply some low-communication protocol for this common
function $q$. This would be the ``resolve the uncertainty first'' approach. We prove (Theorem~\ref{thm:LB}) that resolving the uncertainty can be very expensive (much more so than even the trivial protocol of sending $x$) and hence, this would not be a way to prove Theorem~\ref{thm:one-way}. Instead, we show a path around the inherent
uncertainty to computing the desired function, and this leads to a proof of Theorem~\ref{thm:one-way}. To handle non-product distributions in Theorem~\ref{thm:one-way}, we in particular use a \emph{one-way distributional} variant of the correlated sampling protocol of Braverman and Rao \cite{braverman2011information}. For a high-level overview of the proof of Theorem~\ref{thm:one-way}, we refer the reader to Section~\ref{subsec:overv_prot}.

We now describe our lower bound. Given the upper bound in Theorem~\ref{thm:one-way}, a natural question is whether the dependence on $I(X;Y)$ in the right-hand side of \Cref{eq:main_up_bd_intro} is actually needed. In other words, is it also the case that for non-product distributions, contextual uncertainty can only cause a constant-factor blow-up in communication (for constant error probabilities) ? Perhaps surprisingly, the answer to this question turns out to be negative. Namely, we show that a dependence of the communication in the uncertain setting on $I(X;Y)$ is required.

\begin{theorem}\label{thm:main_lb_non_prod}
There exist a distribution $\mu$ and a function class $\calf \subseteq \owF^\mu_{1,0,\delta}$ such that for every $\epsilon > 0$,
$$ \CCU^{\mu}_{\frac12 - \epsilon}(\calf) \geq 
\Omega( \sqrt{\delta n}) -\log(1/\epsilon). $$
\end{theorem}

In particular, if $\delta$ is any small constant (e.g., $1/5$), then Theorem~\ref{thm:main_lb_non_prod} asserts the existence of a distribution and a class of distance-$\delta$ functions for which the zero-error (one-way) communication complexity in the standard model is a single bit, but under contextual uncertainty, any two-way protocol (with an arbitrary number of rounds of interaction) having a noticeable advantage over random guessing requires $\Omega(\sqrt{n})$ bits of communication! We note that the distribution $\mu$ in Theorem~\ref{thm:main_lb_non_prod} has mutual information $\approx n$, so Theorem~\ref{thm:main_lb_non_prod} rules out improving the dependence on the mutual information in \Cref{eq:main_up_bd_intro} to anything smaller than $\sqrt{I(X;Y)}$. It is an interesting open question to determine the correct exponent of $I(X;Y)$ in \Cref{eq:main_up_bd_intro}.

In order to prove Theorem~\ref{thm:main_lb_non_prod}, the function class $\calf$ will essentially consist of the set of all close-by pairs of parity functions and the distribution $\mu$ will correspond to the \emph{noisy Boolean hypercube}. We are then able to reduce the problem of computing $\calf$
under $\mu$ with contextual uncertainty, to the problem of computing a related function in the standard distributional communication complexity model (i.e., \emph{without} uncertainty) under a related distribution. We then use the \emph{discrepancy method} to prove a lower bound on the communication complexity of the new problem. This task itself reduces to upper bounding the spectral norm of a certain communication matrix. The choice of our underlying distribution then implies a tensor structure for this matrix, which reduces the spectral norm computation to bounding the largest singular value of an explicit family of $4 \times 4 $ matrices. For more details about the proof of Theorem~\ref{thm:main_lb_non_prod}, we refer the reader to Section~\ref{sec:lb_non_prod}.

\paragraph{Contrast with prior work.}

The first works to consider communication with uncertainty in a manner similar to this work were those
of \cite{JS1,GJS:jacm}. Their goal was to model an extreme form of uncertainty, where Alice and Bob do not
have any prior (known) commonality in context and indeed both come with their own ``protocol'' which
tells them how to communicate. So communication is needed even to resolve this uncertainty.
While their setting is thus very broad, the solutions they propose are much slower and typically involve
resolving the uncertainty as a first step.

The later works \cite{JKKS11,HaramatyS,CGMS15} tried to restrict the forms of uncertainty to see when it
could lead to more efficient communication solutions. For instance, Juba et al.~\cite{JKKS11} consider the
compression problem when Alice and Bob do not completely agree on the prior. This introduces some
uncertainty in the beliefs, and they provide fairly efficient solutions by restricting the uncertainty to a
manageable form. Canonne et al.~\cite{CGMS15} were the first to connect this stream of work to
communication complexity, which seems to be the right umbrella to study the broader communication
problems. The imperfectness they study is however restricted to the randomness shared by the communicating
parties, and does not incorporate any other elements. They suggest studying imperfect understanding of the
function being computed as a general direction, though they do not suggest specific definitions, which we in particular do in this work.

\paragraph{Organization} 
In Section~\ref{sec:model}, we carefully develop the uncertain communication complexity model after recalling the standard distributional communication complexity model. In Section~\ref{sec:LB}, we prove the hardness of contextual agreement. In Section~\ref{sec:ow_ub}, we prove our main upper bound (Theorem~\ref{thm:one-way}). In Section~\ref{sec:lb_non_prod}, we prove our main lower bound (Theorem~\ref{thm:main_lb_non_prod}). For a discussion of some intriguing future directions that arise from this work, we refer the reader to the conclusion section \ref{sec:conc}.

%% file: model.tex
\section{The Uncertain Communication Complexity Model}\label{sec:model}

We start by recalling the classical communication complexity model of
Yao~\cite{Yao} and then present our definition and measures.

\subsection{Communication Complexity}\label{sec:comm_comp}

We start with some basic notation.
For an integer $n \in \mathbb{N}$, we denote by $[n]$ the set $\{1,\ldots, n\}$. We use $\log{x}$ to denote a logarithm in base $2$. For two sets $A$ and $B$, we denote
by $A \triangle B$ their symmetric difference.
For a distribution $\mu$, we denote by $x \sim \mu$ the
process of sampling a value $x$ from the distribution $\mu$. Similarly, for a set
${X}$ we denote by $x \sim {X}$ the process of sampling a
value $x$ from the uniform distribution over $X$.
For any event $E$, let $\1(E)$ be the $0$-$1$ indicator of
$E$.
For a probability distribution $\mu$ over $X \times Y$, we denote by $\mu_X$ the
marginal of $\mu$ over $X$. By $\mu_{Y \mid x}$, we denote the conditional
distribution of $\mu$ over $Y$ conditioned on $X = x$.

Given a distribution $\mu$ supported on $X$ and functions $f,g\colon X \to
\Sigma$, we let $\delta_\mu(f,g)$ denote the (weighted and normalized)
Hamming distance between $f$ and $g$, i.e., $\delta_\mu(f,g) \triangleq \Pr_{x \sim
\mu} [f(x) \ne g(x)]$. (Note that this definition extends naturally to
probabilitistic functions $f$ and $g$ -- by letting $f(x)$ and $g(x)$ be
sampled independently.)

We now turn to the definition of {\em communication complexity}.
A more thorough
introduction can be found in \cite{KN97}. Let $f\colon X\times Y\to \zo$ be a
function and Alice and Bob be two parties. A protocol $\Pi$ between Alice and
Bob specifies how and what Alice and Bob communicate given their respective
inputs and communication thus far. It also specifies when they stop and produce an
output (that we  require to be produced by Bob). A protocol is said to be
one-way if it involves a single message from Alice to Bob, followed by Bob
producing the output. The protocol $\Pi$ is said to compute $f$ if for every $(x,y)\in X \times Y$ it holds that
$\Pi(x,y) = f(x,y)$. The communication complexity of $\Pi$ is the number of bits
transmitted during the execution of the protocol between Alice and Bob. The communication
complexity of $f$ is the minimal communication complexity of a protocol computing $f$.

It is standard to relax the above setting by introducing a distribution $\mu$
over the input space $X\times Y$ and requiring the protocol to succeed with high
probability (rather than with probability 1). 
We say that a protocol $\Pi$ $\epsilon$-computes a function $f$ under
distribution $\mu$ if $\delta_{\mu}(\Pi(x,y),f(x,y)) \leq \epsilon$.

\begin{definition}[Distributional Communication Complexity]\label{def:dcc}
Let $f\colon X \times Y \rightarrow \zo$ be a Boolean function and $\mu$ be a
probability distribution over $X \times Y$. The \emph{distributional
communication complexity} of $f$ under $\mu$ with error $\epsilon$, denoted by
$\CC^{\mu}_\epsilon(f)$, is defined as the minimum over all protocols
$\Pi$ that $\epsilon$-compute $f$ over $\mu$, of the
communication complexity of $\Pi$.
The one-way communication complexity $\owCC^\mu_\epsilon(f)$ is defined
similarly by minimizing over one-way protocols $\Pi$.
\end{definition}

We note that it is also standard to provide Alice
and Bob with a shared random string which is independent of $x$, $y$ and $f$. In the distributional communication complexity model, it is a known fact that any protocol with shared randomness can be used to get a protocol
that does \emph{not} use shared randomness without increasing its distributed communication complexity.

In this paper, unless stated otherwise, whenever we refer to a protocol, we think of the input pair $(x,y)$ as coming from a distribution.

\subsection{Uncertain Communication Complexity}

We now turn to the central definition of this paper, namely \emph{uncertain communication complexity}.
Our goal is to understand how Alice and Bob can communicate when the
function that Bob wishes to determine is not known to Alice.
In this setting, we make the functions $g$ (that Bob wants to
compute) and $f$ (Alice's estimate of $g$) explicitly part of the input to
the protocol $\Pi$. Thus, in this setting a protocol $\Pi$ specifies how
Alice with input $(f,x)$ and Bob with input $(g,y)$ communicate, and how
they stop and produce an output.
We say that $\Pi$ computes $(f,g)$ if for every
$(x,y) \in X \times Y$, the protocol outputs $g(x,y)$.
We say that a (public-coin) protocol $\Pi$ $\epsilon$-computes $(f,g)$ over $\mu$ if
$\delta_\mu(g,\Pi) \leq \epsilon$.

Next, one
may be tempted to define the communication complexity of a pair of functions
$(f,g)$ as the minimum over all protocols that compute $(f,g)$ of their maximum communication. But this does not capture the uncertainty! (Rather, a protocol that
works for the pair corresponds to both Alice and Bob knowing both $f$ and $g$.)
To model uncertainty, we have to consider the communication complexity of
a whole class of pairs of functions, from which the pair $(f,g)$ is chosen (in our case by an adversary).

Let $\calf \subseteq \{f\colon X \times Y \to \zo\}^2$ be a family of pairs of Boolean
functions with domain $X\times Y$.
We say that a public-coin protocol $\Pi$ $\epsilon$-computes $\calf$ over $\mu$
if for every $(f,g) \in \calf$, we have that $\Pi$ $\epsilon$-computes
$(f,g)$ over $\mu$. We are now ready to present our main definition.


\begin{definition}[Contextually Uncertain Communication Complexity]\label{def:ucc_general}
  Let $\mu$ be a distribution on $X \times Y$ and
$\calf \subseteq \{f\colon X \times Y \to \zo\}^2$.
  The communication complexity of $\calf$ under contextual uncertainty,
  denoted
  $\CCU^\mu_{\epsilon}(\calf)$, is the minimum over all public-coin protocols $\Pi$
that $\epsilon$-compute $\calf$ over $\mu$, of the maximum communication
complexity of $\Pi$ over $(f,g) \in \calf$, $(x,y)$ from the support of
$\mu$ and settings of the public coins.

As usual, the one-way contextually uncertain communication complexity
$\owCCU^\mu_\epsilon(\calf)$ is defined similarly.
\end{definition}

We remark that while in the standard distributional model of Subsection~\ref{sec:comm_comp}, shared randomness can be assumed \emph{without loss of generality}, this is not necessarily the case in Definition~\ref{def:ucc_general}. This is because in principle, shared randomness can help fool the adversary who is selecting the pair $(f,g) \in \calf$. Also, observe that in the special case where $\calf = \{(f,g)\}$, \Cref{def:ucc_general} boils down
to the standard definition of distributional communication complexity (i.e., \Cref{def:dcc}) for the function $g$, and we thus have
$\CCU^\mu_\epsilon(\{(f,g)\}) = \CC^\mu_\epsilon(g)$.
Furthermore, the uncertain communication complexity is monotone, i.e.,
if $\calf \subseteq \calf'$ then
$\CCU^\mu_\epsilon(\calf)
\leq \CCU^\mu_\epsilon(\calf')$. Hence, we conclude that
$\CCU^\mu_\epsilon(\calf) \geq \max_{\{g \;|\; \exists f {\rm ~s.t.~} (f,g) \in \calf\}}
\{\CC^\mu_\epsilon(g)\}$.

In this work, we attempt to identify a setting under which the
above lower bound can be matched. If the set of functions $\Gamma(g) = \{f \;|\;
(f,g) \in \calf\}$ is not sufficiently informative about $g$, then it seems
hard to conceive of settings where Alice can do non-trivially well. We thus
pick a simple and natural restriction on $\Gamma(g)$, namely, that it contains
functions that are close to $g$ (in $\delta_\mu$-distance). This leads us to
our main target classes. For
parameters $k,\epsilon,\delta>0$, define the sets of pairs of functions
\begin{align*}
  \calf_{k,\epsilon,\delta} \triangleq \{(f,g) \;|\; \delta_\mu(f,g) \leq \delta
  ~\& ~ \CC^\mu_\epsilon(f), \CC^\mu_\epsilon(g) \leq k\}
\end{align*}
and 
\begin{align*}
  \owF_{k,\epsilon,\delta} \triangleq \{(f,g) \;|\; \delta_\mu(f,g) \leq \delta
  ~\& ~ \owCC^\mu_\epsilon(f), \owCC^\mu_\epsilon(g) \leq k\}.
\end{align*}

In words, $\calf_{k,\epsilon,\delta}$ (resp.\ $\owF_{k,\epsilon,\delta}$) considers all possible functions $g$
with communication complexity (resp.\ one-way communication complexity) at most $k$ with Alice being
roughly under all possible uncertainties within
distance $\delta$ of Bob.\footnote{For the sake of symmetry, we insist that
$\CC^\mu_\epsilon(f) \leq k$ (resp.\ $\owCC^{\mu}_{\epsilon}(f) \leq k$). We need not have insisted on it but since the other conditions anyhow imply that
$\CC^\mu_{\epsilon+\delta}(f) \leq k$ (resp.\ $\owCC^{\mu}_{\epsilon+\delta}(f) \leq k$), we decided to include this stronger
condition for aesthetic reasons.}

It is clear that $\owCC^\mu_\epsilon(\owF_{k,\epsilon,\delta}) \geq	k$.
Our first main result, Theorem~\ref{thm:one-way}, gives an upper bound on this quantity, which in the particular case of product distributions is comparable to $k$
(up to a constant factor increase in the error and communication complexity). In Theorem~\ref{thm:LB} we show that a naive strategy that
attempts to reduce the uncertain communication problem to a ``function agreement problem'' (where Alice and Bob agree on a function $q$ that is close to
$f$ and $g$ and then use a protocol for $q$) cannot work. Furthermore, our second main result, Theorem~\ref{thm:main_lb_non_prod}, shows that for general non-product distributions, $\CC^\mu_\epsilon(\owF_{k,\epsilon,\delta})$ can be much larger than $k$. More precisely, we construct a function class along with a distribution $\mu$ for which the one-way communication complexity in the standard model is a single bit whereas, under contextual uncertainty, the two-way communication complexity is at least $\Omega(\sqrt{n})$!


%% file: hardag.tex
\newcommand{\AD}{\textsc{Agreement-Distillation}}
\renewcommand{\F}{\mathcal{F}}
\newcommand{\Ent}{\mathsf{Ent}} 

\section{Hardness of Contextual Agreement} \label{sec:LB}

In this section, we show that even if both $f$ and $g$ have small one-way
distributional  communication complexity on some distribution $\mu$, agreeing on a $q$ such that $\delta_{\mu}(q,f)$ is small takes
communication that is roughly the size of the bit representation of $f$ (which
is exponential in the size of the input). Thus, agreeing on $q$ before
simulating a protocol for $q$ is exponentially costlier than even the trivial
protocol where Alice sends her input $x$ to Bob.
Formally, we consider the following communication problem:

\begin{definition}[$\agree_{\delta,\gamma}(\calf)$]
For a family of pairs of functions $\calf \subseteq \{f\colon X \times Y \to \zo\}^2$, the
{$\calf$-agreement} problem with parameters $\delta,\gamma \geq 0$ is the communication problem where Alice gets $f$ and Bob gets
$g$ such that $(f,g) \in \calf$ and their goal is for Alice to output $q_A$ and
Bob to output $q_B$ such that $\delta(q_A,f),\delta(q_B,g) \leq \delta$ and $\Pr[q_A = q_B] \geq \gamma$.

\end{definition}

Somewhat abusing notation, we will use
$\agree_{\delta,\gamma}(\D)$ to denote
the distributional problem where $\D$ is a distribution on
$\{f\colon X \times Y \to \zo\}^2$ and the goal now is to get agreement with
probability $\delta$ over the randomness of the protocol and the input.

If the agreement problem could be solved with low communication for the family $\calf_{k,\delta,\epsilon}$
as defined at the end of Section~\ref{sec:model}, then this would turn into a natural protocol for
$\CCU(\calf_{k,\epsilon',\delta'})$ for some positive $\epsilon'$ and $\delta'$ as well. Our following theorem proves that agreement is a huge overkill.

\begin{theorem}\label{thm:LB}
For every $\delta,\delta_2>0$, there exists $\alpha>0$
and a family $\calf \subseteq \calf_{0,0,\delta}$ such that
for every $\gamma > 0$, it holds that
$\CC(\agree_{\delta_2,\gamma}(\calf)) \geq \alpha |Y| - \log(1/\gamma)$.
\end{theorem}

In words, Theorem~\ref{thm:LB} says that there is a family of pairs of functions supported on functions
of {\em zero}  communication complexity (with zero error) for which agreement takes communication
polynomial in the size of the domain of the functions. Note that this is
exponentially larger than the trivial communication complexity for any
function $g$, which is at most $\min\{1+ \log|Y|, \log |X|\}$.

We stress that while an agreement lower bound for zero communication functions may feel a lower bound for a toy problem,
a lower bound for this setting is inherent in any separation between agreement complexity for $\calf$
and communication complexity with uncertainty for $\calf$. To see this, note that given any input to
the $\CCU(\calf)$ problem, Alice and Bob can execute any protocol for $\CCU(\calf)$ pinning down the value of
the function to be computed with high probability and low communication. If one considers the remaining
challenge to agreement, it comes from a zero communication problem.

Our proof of Theorem \ref{thm:LB} uses a lower bound on
the communication complexity of \emph{agreement distillation (with imperfectly shared randomness)} problem
defined in \cite{CGMS15}, who in turn rely on a lower bound for randomness
extraction from correlated sources due to 
Bogdanov and Mossel~\cite{BogdanovMossel}.

We describe their problem below and the result that we use. We note that their
context is slightly different and our description below is a reformulation.
First, we define the notion of $\rho$-perturbed sequences of bits.
A pair of bits $(a,b)$ is said to be a pair of $\rho$-perturbed uniform bits if $a$ is
uniform over $\{0,1\}$, and $b = a$ with probability $1-\rho$ and $b \ne a$
with probability $\rho$. A pair of sequences of bits $(r,s)$ is said to be
$\rho$-perturbed if $r = (r_1,\ldots,r_n)$ and $s = (s_1,\ldots,s_n)$ and
each coordinate pair $(r_i,s_i)$ is a $\rho$-perturbed uniform pair drawn
independently of all other pairs. For a random variable $W$, we define
its min-entropy as $H_\infty(w)
\triangleq \min_{w \in \mathsf{supp}(W)} \{- \log(\Pr[W = w]\}$.

\begin{definition}[$\AD^k_{\gamma,\rho}$]
In this problem, Alice and Bob get as inputs $r$ and $s$, where $(r,s)$ form a
$\rho$-perturbed sequence of bits.
Their goal is to communicate deterministically
and produce as outputs $w_A$ (Alice's output) and $w_B$ (Bob's output) with the
following
properties:
(i) $H_\infty(w_A), H_\infty(w_B) \geq k$ and (ii) $\Pr_{(r,s)} [w_A = w_B] \geq
\gamma$.
\end{definition}

\begin{lemma}[{\protect \cite[Theorem 2]{CGMS15}}]
\label{lem:cgms}
For every $\rho>0$, there exists $\epsilon > 0$ such that for every $k$ and
$\gamma$, it holds that every deterministic protocol $\Pi$ that computes
$\agree^k_{\gamma,\rho}$ has communication complexity at least $\epsilon k -
\log 1/\gamma$.
\end{lemma}

We note that while the agreement distillation problem is very similar to our agreement problem, there
are some syntactic differences. We are considering pairs of functions with low
communication complexity, whereas the agreement-distillation problem considers arbitrary
random  sequences. Also, our output criterion is proximity to the input
functions, whereas in the agreement-distillation problem, we need to produce
high-entropy outputs. Finally, we want a lower bound for our agreement
problem when Alice and Bob are allowed to share perfect randomness while the agreement-distillation bound only holds for deterministic protocols.
Nevertheless, we are able to reduce to their setting quite easily as we will
see shortly.

Our proof of Theorem~\ref{thm:LB} uses the standard Chernoff-Hoeffding tail
inequality on random variables that we include below. Denote $\exp(x)\triangleq
e^x$, where $e$ is the base of the natural logarithm.
\begin{proposition}[Chernoff bound]\label{prop:chernoff}
  Let $X=\sum_{i=1}^nX_i$ be a sum of identically distributed independent random
  variables $X_1,\dots,X_n\in \{0,1\}$. Let $\mu = \Ex[X] = \sum_{i=1}^n\Ex[X_i]$. It
  holds that for $\delta \in (0,1)$,
  \begin{align*}
    \Pr[X < (1 - \delta)\mu] \leq \exp\left(-\delta^2\mu/2\right)
  \end{align*}
  and
  \begin{align*}
    \Pr[X > (1 + \delta)\mu] \leq \exp\left(-\delta^2\mu/3\right),
  \end{align*}
  and for $a >0$,
  \begin{align*}
    \Pr[X > \mu + a] \leq \exp(-2a^2/n)
  \end{align*}
\end{proposition}

\begin{proof}[Proof of Theorem~\ref{thm:LB}]
We prove the theorem for $\alpha < \delta/6$, in which case we may assume
$\gamma > \exp(-\delta |Y| /6)$ since otherwise the right-hand side is
non-positive.

Let $\F_B$ denote the set of functions that depend only on Bob's inputs, i.e., $f\in \F_B$ if there
exists $f'\colon Y \to \zo$ such that $f(x,y) = f'(y)$ for all $x,y$.
Our family $\calf$ will be a subset of $\F_B \times \F_B$, the subset that contains functions that are at most $\delta |Y|$ apart.
$$\calf \triangleq \{(f,g) \in \F_B \times \F_B ~|~ \delta(f,g)  \leq \delta\}.$$
It is clear that communication complexity of every function in the support of $\calf$ is zero, with zero error (Bob can compute it on his
own) and so $\calf \subseteq \calf_{0,0,\delta}$.
So it remains to prove a lower bound on $\CC(\agree_{\delta_2,\gamma}(\calf))$.

We prove our lower bound by picking a
distribution $\D_\rho$
supported mostly on $\calf$ and by giving a lower bound
on $\CC(\agree_{\delta_2,\gamma}(\D_\rho))$.
Let $\rho = \delta/2$.
The distribution $\D_\rho$ is a simple one. It samples $(f,g)$ as
follows. The function $f$ is drawn uniformly at random from
$\calf_B$. Then, $g$ is chosen to be a ``$\rho$-perturbation'' of $f$,
namely for every $y \in Y$, $g'(y)$ is chosen to be equal to $f(x,y)$ with probability $1-\rho$ and $1-f(x,y)$ with probability
$\rho$. For every $x \in X$, we now set $g(x,y) = g'(x,y)$.

By the Chernoff bound (see Proposition~\ref{prop:chernoff}), we have
that
$\Pr_{(f,g) \sim \D_\rho} [\delta(f,g) > \delta] = \exp(-\rho|Y|/3) \leq
\gamma$.
So with overwhelmingly high probability, $\D_\rho$ draws elements from
$\calf$.
In particular, if some protocol solves $\agree_{\delta_2,\gamma}(\calf)$,
then it would also solve $\agree_{\delta_2,2\gamma}(\D_\rho)$.

We thus need to show a lower bound on the communication complexity
of $\agree_{\delta_2,2\gamma}(\D_\rho)$. We now note that since this is a
distributional problem, by Yao's min-max principle, if there is randomized
protocol to solve
$\agree_{\delta_2,2\gamma}(\D_\rho)$ with $C$ bits of communication, then there
is also a deterministic protocol for the same problem and with the same complexity.
Thus, it suffices to lower bound the deterministic communication complexity of $\agree_{\delta_2,2\gamma}(\D_\rho)$.
Claim~\ref{clm:reduction} shows that any such protocol gives a deterministic protocol
for $\AD$ with $k = \Omega_{\delta_2}(|Y|)$. Combining this with
Lemma~\ref{lem:cgms} gives us the desired lower bound on
$\CC(\agree_{\delta_2,2\gamma}(\D_\rho))$ and hence on
$\CC(\agree_{\delta_2,\gamma}(\calf))$.
\end{proof}

\begin{claim}\label{clm:reduction}
Every protocol for $\agree_{\delta_2,\gamma}(\D_\rho)$ is also a protocol
for $\AD^k_{\gamma,\rho}$ for $k = (1 - h(\delta_2))|Y|$, where
$h(\cdot)$ is the binary entropy function given by $h(x) = - x \log x -
(1-x) \log (1-x)$.
\end{claim}

\begin{proof}
Suppose Alice and Bob are trying to solve $\AD^k_{\gamma,\rho}$. They can
sample $\rho$-pertubed strings $(r,s) \in \{0,1\}^{|Y|}$ and interpret them
as functions $f',g' \colon Y \to \zo$ or equivalently as functions
$(f,g) \sim \D_\rho$. They can now simulate the protocol for
$\agree_{\delta_,\gamma}(f,g)$ and output $q_A$ and $q_B$. By definition of
$\agree$, we have $q_A = q_B$ with probability at least $\gamma$. So it suffices to show that $H_\infty(q_A), H_\infty(q_B) \geq k$.
But this is obvious since any function $q_A$ is output only if $\delta(f,q_A)
\leq \delta_2$ and we have that $|\{f \;|\; \delta(f,q_A) \leq \delta\}| \leq
2^{h(\delta_2) |Y|}$. Since the probability of sampling $f$ for any $f$
is at most
$2^{-|Y|}$, we have that the probability of outputting $q_A$ for any $q_A$
is at most $2^{-(1-h(\delta_2))|Y|}$. In other words,
$H_\infty(q_A) \geq (1-h(\delta_2)) |Y|$. Similarly, we can lower bound
$H_\infty(q_B)$ and thus we have that the outputs of the protocol for $\agree$
solve $\AD$ with $k = (1-h(\delta_2))|Y|$.
\end{proof}

%% file: one-way-ub.tex
\section{One-way Communication with Contextual Uncertainty}\label{sec:ow_ub}


In this section, we prove Theorem~\ref{thm:one-way}. We start with a high-level description of the protocol.

\subsection{Overview of Protocol}\label{subsec:overv_prot}
Let $\mu$ be a distribution over an input space $X\times Y$. For any function $s\colon X \times Y \rightarrow \zo$
and any $x \in X$, we define the {\em restriction} of $s$ to $x$ to be the function $s_x\colon Y \rightarrow \zo$ given by $s_x(y) = s(x,y)$ for any $y \in Y$. 

We now give a high-level overview of the protocol. First, we consider the particular case of Theorem~\ref{thm:one-way} where $\mu$ is a product distribution, i.e., $\mu = \mu_X \times \mu_Y$. Note that in this case, $I(X;Y) = 0$ in the right-hand side of \Cref{eq:main_up_bd_intro}. We will handle the case of general (not necessarily product) distributions later on.

The general idea is that given inputs $(f,x)$, Alice can determine the restriction $f_x$, and she will
try to describe it to Bob. For most values of $x$, $f_x$ will be close (in $\delta_{\mu_Y}$-distance)
to the function $g_x$. Bob will try to use the (yet unspecified) description given by 
Alice in order to determine some function $B$ that is close to $g_x$. If he succeeds in doing so, he can output 
$B(y)$ which would equal $g_x(y)$ with high probability over $y$.

We next explain how Alice will describe $f_x$, and how Bob will determine some function
$B$ that is close to $g_x$ based on Alice's description. For the first part, we let Alice and Bob use shared randomness in order to sample
$y_1,\ldots,y_m$, where the $y_i$'s are drawn independently with $y_i \sim \mu_Y$, and $m$ is a parameter
to be chosen later.
Alice's description of $f_x$ will then be $(f_x(y_1),\ldots,f_x(y_m)) \in \{0,1\}^m$.
Thus, the length of the communication is $m$ bits and we need to show that
setting $m$ to be roughly $O(k)$ suffices.
Before we explain this, we first need to specify what Bob does with Alice's message.

As a first cut, let us consider the following natural strategy: Bob picks an $\tilde{x} \in X$ such that $g_{\tilde{x}}$ is close to $f_x$ on $y_1,\ldots,y_m$, and sets $B = g_{\tilde{x}}$.  It is clear that if $\tilde{x} = x$, then $B = g_{\tilde{x}} = g_x$, and for every $y \in \mu_Y$, we would have $B(y) = g_x(y)$. Moreover, if $\tilde{x}$ is such that $g_{\tilde{x}}$ is close to $g_x$ (which is itself close to $f_x$), then $B(y)$ would now equal $g_x(y)$ with high probability.  It remains to deal with $\tilde{x}$ such that $g_{\tilde{x}}$ is far from $g_x$. Note that if we first fix any such $\tilde{x}$ and then sample $y_1,\ldots,y_m$, then with high probability, we would reveal that $g_{\tilde{x}}$ is far from $g_x$. This is because $g_x$ is close to $f_x$, so $g_{\tilde{x}}$ should also be far from $f_x$. However, this idea alone cannot deal with all possible $\tilde{x}$ --- using a naive union bound over all possible $\tilde{x} \in X$ would require a failure probability of $1/|X|$, which would itself require setting $m$ to be roughly $\log |X|$. Indeed, smaller values of $m$ should not suffice since we have not yet used the fact that $\CC^\mu_\epsilon(g) \leq k$ --- but we do so next.

Suppose that $\Pi$ is a one-way protocol with $k$ bits of communication. Then, note that Alice's message partitions $X$ into $2^k$ sets, one corresponding to each message.  Our modified strategy for Bob is to let him pick a representative $x$ from each set in this partition, and then set $B = g_{\tilde{x}}$ for an $\tilde{x}$ among the representatives for which $g_{\tilde{x}}$ and $f$ are the closest on the samples $y_1,\ldots,y_m$.  A simple analysis shows that the $g_x$'s that lie inside the same set in this partition are close, and thus, if we pick $\tilde{x}$ to be the representative of the set containing $x$, then $g_{\tilde{x}}$ and $f_x$ will be close on the sampled points.  For an other representative, once again if $g_{\tilde{x}}$ is close to $g_x$, then $g_{\tilde{x}}(y)$ will equal $g_x(y)$ with high probability. For a representative $\tilde{x}$ such that $g_{\tilde{x}}$ is far from $g_x$ (which is itself close to $f_x$), we can proceed as in the previous paragraph, and now the union bound works out since the total number of representatives is only $2^k$.\footnote{We note that a similar idea was used in a somewhat different context by \cite{YossefJKS02} (following on \cite{KremerNR99}) in order to characterize one-way communication complexity of any function under product distributions in terms of its VC-dimension.}

We now turn to the case of general (not necessarily product) distributions. In this case, we would like to run the above protocol with $y_1, y_2, \dots, y_m$ sampled independently from $\mu_{Y | x}$ (instead of $\mu_Y$). Note that Alice knows $x$ and hence knows the distribution $\mu_{Y | x}$. Unfortunately, Bob does not know $\mu_{Y | x}$; he only knows $\mu_Y$ as a ``proxy'' for $\mu_{Y | x}$. While Alice and Bob cannot jointly sample such $y_i$'s without communicating (as in the product case), they can still run the correlated sampling protocol of \cite{braverman2011information} in order to agree on such samples while communicating at most $O(m \cdot I(X;Y))$ bits. The original correlated sampling procedure of \cite{braverman2011information} inherently used multiple rounds of communication, but we are able in our case to turn it into a one-way protocol by leveraging the fact that our setup is distributional (see Subsection~\ref{subsec:corr_samp} for more details).

The outline of the rest of this section is the following. In Subsection~\ref{subsec:corr_samp}, we describe the properties of the correlated sampling procedure that we will use. In Subsection~\ref{subsec:proof_th_ub}, we give the formal proof of Theorem~\ref{thm:one-way}.

\subsection{Correlated Sampling}\label{subsec:corr_samp}

We start by recalling two standard notions from information theory. Given two disributions $P$ and $Q$, the \emph{KL divergence} between $P$ and $Q$ is defined as $D(P || Q) \triangleq \Ex_{u \sim P}[\log(P(u)/Q(u))]$. Given a joint distribution $\mu$ of a pair $(X,Y)$ of random variables with $\mu_X$ and $\mu_Y$ being the marginals of $\mu$ over $X$ and $Y$ respectively, the \emph{mutual information} of $X$ and $Y$ is defined as $I(X;Y) \triangleq D( \mu || \mu_X \mu_Y)$.

The following lemma summarizes the properties of the \emph{correlated sampling} protocol of \cite{braverman2011information}.
\begin{lemma}[\cite{braverman2011information}]\label{lem:BR_prot}
Let Alice be given a distribution $P$ and Bob be given a distribution $Q$ over a common universe $\mathcal{U}$. There is an interactive public-coin protocol that uses an expected
$$ D(P || Q) + 2\log(1/\epsilon) + O(\sqrt{D(P || Q)}+1) $$
bits of communication such that at the end of the protocol:
\begin{itemize}
\item Alice outputs an element $a$ distributed according to $P$.
\item Bob outputs an element $b$ such that for each $u \in \mathcal{U}$, $\Pr[b=u ~ | ~ a=u] > 1-\epsilon$.
\end{itemize}
Moreover, the message that Bob sends to Alice in any given round consists of a single bit indicating if the protocol should terminate or if Alice should send the next message.
\end{lemma}

We point out that in general, the correlated sampling procedure in Lemma~\ref{lem:BR_prot} can take more than one round of communication. This is because initially, neither Alice nor Bob knows $D(P || Q)$ and they will need to interactively ``discover'' it. In our case, we will be using correlated sampling in a ``distributional setup''. It turns out that this allows us to use a \emph{one-way version} of correlated sampling which is described in Lemma~\ref{lem:one_w_cor_samp} below.

\begin{lemma}\label{lem:one_w_cor_samp}
Let $\mu$ be a distribution over $(x,y)$ with marginal $\mu_X$ over $x$, and assume that $\mu$ is known to both Alice and Bob. Fix $\epsilon > 0$ and let Alice be given $x \sim \mu_X$. There is a \emph{one-way} public-coin protocol that uses at most
$$O( m \cdot I(X;Y)/\epsilon + \log(1/\epsilon)/\epsilon)$$
bits of communication such that with probability at least $1-\epsilon$ over the public coins of the protocol and the randomness of $x$, Alice and Bob agree on $m$ samples $y_1,y_2,\dots, y_m \text{ i.i.d } \sim \mu(Y | x)$ at the end of the protocol.
\end{lemma}

\begin{proof}
When $x$ is Alice's input, we can consider running the protocol in Lemma~\ref{lem:BR_prot} on the distributions $P \triangleq \prod_{i=1}^m \mu(Y_i | x)$ and $Q \triangleq \prod_{i=1}^m \mu(Y_i)$ and with error parameter $\epsilon/2$. Let $\Pi$ be the resulting protocol transcript. The expected communication cost of $\Pi$ is at most
\begin{align}\label{eq:markov_ex}
\Ex_{x \sim \mu_X}[ O(D(P || Q)) + O(\log(1/\epsilon))] &= O(\Ex_{x \sim \mu_X}[ D(P || Q)]) + O(\log(1/\epsilon)) \notag\\ 
&= O(m \cdot I(X;Y)) + O(\log(1/\epsilon)),
\end{align}
where the last equality follows from the fact that
\begin{align*}
\Ex_{x \sim \mu_X}[ D(P || Q)] &= \Ex_{x \sim \mu_X}\bigg[ \Ex_{y_1 | x, \dots,y_m | x} \bigg[ \log \bigg(\frac{\prod_{i=1}^m \mu(y_i | x)}{\prod_{i=1}^m \mu(y_i)} \bigg) \bigg] \bigg]\\ 
&= \displaystyle\sum\limits_{i=1}^m \Ex_{x \sim \mu_X}\bigg[ \Ex_{y_1 | x, \dots,y_m | x} \bigg[ \log \bigg(\frac{\mu(y_i | x)}{\mu(y_i)} \bigg) \bigg] \bigg]\\ 
&= \displaystyle\sum\limits_{i=1}^m \Ex_{x \sim \mu_X}\bigg[ \Ex_{y_i | x} \bigg[ \log \bigg(\frac{\mu(y_i | x)}{\mu(y_i)} \bigg) \bigg] \bigg]\\ 
&= \displaystyle\sum\limits_{i=1}^m \Ex_{(x,y) \sim \mu}\bigg[ \log \bigg(\frac{\mu(y | x)}{\mu(y)} \bigg) \bigg]\\ 
&= m \cdot I(X;Y).
\end{align*}
By Markov's inequality applied to (\ref{eq:markov_ex}), we get that with probability at least $1-\epsilon/2$, the length of the transcript $\Pi$ is at most
$$ \ell \triangleq O(m \cdot I(X;Y) /\epsilon) + O(\log(1/\epsilon)/\epsilon).$$

Conditioned on the event $E$ that the length of $\Pi$ is at most $\ell$ bits, the total number of bits sent by Alice to Bob is also at most $\ell$.

Note that Lemma~\ref{lem:BR_prot} guarantees that each message of Bob in $\Pi$ consists of a single bit indicating if the protocol should terminate or if Alice should send the next message. Hence, Bob's messages do not influence the actual bits sent by Alice; they only determine how many bits are sent by her.

In the new one-way protocol $\Pi'$, Alice sends to Bob, in a single shot, the first $\ell$ bits that she would have sent him in protocol $\Pi$ if he kept refusing to terminate. Upon receiving this message, Bob completes the simulation of protocol $\Pi$. The error probability of the new protocol $\Pi'$ is the probability that either Alice did not send enough bits or that the protocol $\Pi$ makes an error, which by a union bound is at most
$$ \Pr[\overline{E}] + \epsilon/2 \le \epsilon/2 + \epsilon/2 = \epsilon$$
where $\overline{E}$ denotes the complement of event $E$.
\end{proof}

\subsection{Proof of Theorem~\ref{thm:one-way}}\label{subsec:proof_th_ub}

Recall that in the contextual setting, Alice's input is $(f,x)$ and Bob's input is $(g,y)$, where $(f,g)\in\owF^\mu_{k,\epsilon,\delta}$ and $(x,y) \sim \mu$. Let $\Pi$ be the one-way protocol for $g$ in the standard setting that shows that $\owCC^\mu_\epsilon(g) \leq k$. Note that $\Pi$ can be
described by an integer $L \leq 2^k$ and functions
$\pi\colon X \to [L]$ and $\{B_i \colon Y \to \zo \}_{i \in [L]}$, such that Alice's message on input $x$ is $\pi(x)$, and Bob's output on message $i$ from Alice and on input $y$ is $B_i(y)$.
We use this notation below. We also set the parameter $m = \Theta\big(c (k+\log(1/\theta))/\theta^2 \big)$, which is chosen such that $2^k \cdot \exp(-\theta^2m/75)\leq 2\theta/5$.

\paragraph{The protocol.} 
Algorithm
\ref{Fig:Protocol} describes the protocol we employ in the contextual setting. Roughly speaking, the protocol works as follows. First, Alice and Bob run the one-way correlated sampling procedure given by Lemma~\ref{lem:one_w_cor_samp} in order to sample $y_1,y_2\dots,y_{m} \text{ i.i.d. } \sim \mu_{Y | x}$. Then, Alice sends the sequence  $(f_x(y_1),\ldots,f_x(y_m))$ to Bob. Bob enumerates over $i \in [L]$ and counts the fraction of $z \in \{y_1,\ldots,y_m\}$ for which $B_i(z)\neq f_x(z)$. For the index $i$ which minimizes this fraction, Bob outputs $B_i(y)$ and halts. 

  \begin{algorithm}[H]
    \caption{The protocol that handles contextual
      uncertainty} \label{Fig:Protocol}%
    \textbf{The setting:} Let $\mu$ be a probability distribution over a message
    space $X\times Y$. Alice and Bob are given functions $f$ and $g$, and inputs
    $x$ and $y$, respectively, where $(f,g)\in\owF^\mu_{k,\epsilon,\delta}$ and
    $(x,y) \sim \mu$.

    \textbf{The protocol:}
    \begin{enumerate}
    \item Alice and Bob run one-way correlated sampling with error parameter set to $(\theta/10)^2$ in order to sample $m$ values $Z = \{ y_1, y_2, \ldots,
      y_m \} \subseteq Y$ each sampled independently according to $\mu_{Y | x}$.
    \item Alice sends $\{ f_x(y_i) \}_{i \in [m]}$ to Bob.
    \item 
      For every $i \in [L]$, Bob computes
      $\err_i \triangleq \frac{1}{m} \sum_{j = 1}^m \1( B_i(y_j) \neq f_x(y_j))$.
      Let $i_{\min} \triangleq \mathop{\mathrm{argmin}}_{i\in [L]}\{\err_i\}$.
      Bob outputs $B_{i_{\min}}(y)$ and halts. 
    \end{enumerate}
\end{algorithm}

\paragraph{Analysis.} Observe that by Lemma~\ref{lem:one_w_cor_samp}, the correlated sampling procedure requires $O( m \cdot I(X;Y)/\theta^2 + \log(1/\theta)/\theta^2)$ bits of communication. Thus, the total communication of our protocol is at most
\begin{align*}
O( m \cdot I(X;Y)/\theta^2 + \log(1/\theta)/\theta^2) + m = \frac{c \left(k+\log\left(\frac{1}{\theta}\right)\right)}{\theta^2} \cdot \left(1+\frac{I(X;Y)}{\theta^2}\right) \text{ bits }
\end{align*}
for some absolute constant $c$, as promised. The next lemma establishes the correctness of the protocol.

\begin{lemma}\label{le:ow_ub_analysis}
$\Pr_{\Pi, (x,y)\sim\mu}\left[B_{i_{\min}}(y) \ne g(x,y)\right] \le \epsilon + 2\delta + \theta$.
\end{lemma}

\begin{proof}
We start with some notation. For $x \in X$, let $\delta_x \triangleq \delta_{\mu_{Y | x}}(f_x,g_x)$ and let $\epsilon_x \triangleq \delta_{\mu_{Y | x}}(g_x,B_{\pi(x)})$.
Note that by definition, $\delta = \Ex_{x \sim \mu_X}[\delta_x]$ and
$\epsilon = \Ex_{x \sim \mu_X}[\epsilon_x]$.
For $i \in [L]$, let $\gamma_{i,x} \triangleq \delta_{\mu_{Y | x}}(f_x,B_i)$.
Note that by the triangle inequality,
\begin{align}\label{eq:triangle}
  \gamma_{\pi(x),x} = \delta_{\mu_{Y | x}}(f_x,B_{\pi(x)}) \leq \delta_x + \epsilon_x.
\end{align}

In what follows, we will analyze the probability that $B_{i_{\min}}(y) \ne g(x,y)$
by analyzing the estimate $\err_i$ and the index $i_{\min}$ computed in the above protocol. Note that $\err_i = \err_i(x)$ computed above attempts to estimate $\gamma_{i,x}$, and that both $\err_i$ and $i_{\min}$ are
functions of $x$.

Note that Lemma~\ref{lem:one_w_cor_samp} guarantees that correlated sampling succeeds with probability at least $1-\theta^2/100$. Henceforth, we condition on the event that correlated sampling succeeds (we will account for the event where this does happen at the end). By the Chernoff bound,
we have for every $x$ and $i \in [L]$
\begin{align*}
  \Pr_{y_1,\ldots,y_m\sim\mu_{Y | x}} \left[ \left|\gamma_{i,x} - \err_i\right| >
    \frac{\theta}{5} \right] \leq \exp\left(-\frac{\theta^2\cdot m}{75}\right).
\end{align*}
By a union bound, we have for every $x\in X$, 
\begin{align*}
  \Pr_{y_1,\ldots,y_m\sim\mu_{Y | x}} \left[ \exists i \in [L] \mbox{ s.t. } \left|\gamma_{i,x}
      - \err_i\right| > \frac{\theta}{5} \right] \leq L \cdot
  \exp\left(-\frac{\theta^2\cdot m}{75}\right)
  \leq \frac{2\theta}{5},
\end{align*}
where the last inequality follows from our choice of $m = \Theta\big(c \cdot (k+\log(1/\theta))/\theta^2 \big)$.

Now assume that for all $i \in [L]$, we have 
that $|\gamma_{i,x} - \err_i| \leq \theta/5$, which we refer to below as the ``Good Event". Then, for $i_{\min}$, we have 
\begin{align*}
\gamma_{i_{\min},x} & \leq \err_{i_{\min}} + \theta/5 \tag{since we assumed the
Good Event} \\
& \leq \err_{\pi(x)} + \theta/5 \tag{by definition of $i_{\min}$}\\
& \leq \gamma_{\pi(x),x} + 2\theta/5 \tag{since we assumed the Good Event}\\
& \leq \delta_x + \epsilon_x + 2\theta/5. \tag{By \Cref{eq:triangle}}
\end{align*}
Let $W \subseteq \{0,1\}^n$ be the set of all $x$ for which correlated sampling succeeds with probablity at least $1-\theta/10$ (over the internal randomness of the protocol). By Lemma~\ref{lem:one_w_cor_samp} and an averaging argument, $\Pr_{x \sim \mu_X}[ x \notin W] \le \theta/10$. Thus,
\begin{align*}
\Pr_{\Pi, (x,y)\sim\mu}\left[B_{i_{\min}}(y) \ne f(x,y)\right] &\le \Ex_{x \sim \mu_X | x \in W}\bigg[ \Pr_{\Pi, y \sim \mu_{Y | x}}\left[B_{i_{\min}}(y) \ne f(x,y)\right] \bigg] + \theta/10\\ 
&\le \Ex_{x \sim \mu_X | x \in W}\bigg[ \Pr_{y_1,\ldots,y_m, y \sim \mu_{Y | x}}\left[B_{i_{\min}}(y) \ne f(x,y)\right] \bigg] + \theta/5\\ 
&= \Ex_{x \sim \mu_X | x \in W}\big[ \gamma_{i_{\min},x}  \big] + \theta/5\\ 
&\le \Ex_{x \sim \mu_X | x \in W}\big[ \delta_x + \epsilon_x \big] + 3\theta/5\\ 
&\le  \Ex_{x \sim \mu_X}\big[ \delta_x + \epsilon_x \big] + \theta\\ 
&= \delta+\epsilon + \theta
\end{align*}
where the third inequality follows from the fact that the Good Event occurs with probability at least $1-2\theta/5$, and from the corresponding upper bound on $\gamma_{i_{\min},x}$. The other inequalities above follow from the definition of the set $W$ and the fact that $\Pr_{x \sim \mu_X}[ x \notin W] \le \theta/10$. Finally, since $\delta(f,g) \leq \delta$, we have that Bob's output does not equal
$g(x,y)$ (which is the desired output) with probability at most $\epsilon + 2\delta + \theta$.
\end{proof}

%% file: two-way-lb.tex
\section{Lower Bound for Non-Product Distributions}\label{sec:lb_non_prod}
In this section, we prove Theorem~\ref{thm:main_lb_non_prod}. We start by defining the class of function pairs and distributions that will be used. Consider the parity functions on subsets of bits of the string $x \oplus y \in \zo^n$. Specifically, for every $S \subseteq [n]$, let $f_S\colon\zo^n\times\zo^n\to\zo$ be defined as $f_S(x,y) \triangleq \oplus_{i \in S} (x_i \oplus y_i)$.
Let $q = q(n) > 0$ and define \begin{equation} \calf_q \triangleq \{(f_S,f_T) : |S \triangle T| \leq q\cdot n\}. \end{equation}

Next, we define a probability distribution $\mu_p$ on $\zo^n \times \zo^n$ where $p=p(n)$. We do so by giving a procedure to sample according to $\mu_p$. To sample a pair $(x,y) \sim \mu_p$, we draw $x \in_R \{0,1\}^n$ and let $y$ be a $p$-noisy copy of $x$, i.e., $y \sim N_{p}(x)$. Here, $N_p(x)$ is the distribution on $\zo^n$ that outputs $y \in \zo^n$ such that, independently, for each $i \in [n]$, $y_i = 1-x_i$ with probability $p$, and $y_i = x_i$  with probability $1-p$. In other words, $\mu_p(x,y) = 2^{-n} \cdot p^{|x \oplus y|} \cdot (1-p)^{n-|x \oplus y|}$ for every $(x,y) \in \zo^n \times \zo^n$. Here, $|z|$ denotes the Hamming weight of $z$, for any $z \in \zo^n$.

We will prove Lemmas \ref{le:two_way_ubs} and \ref{le:two_way_lb} below about the function class $\calf_q$ and the distribution $\mu_p$.

\begin{lemma}\label{le:two_way_ubs}
For any $p = p(n)$ and $q = q(n)$, it holds that $\calf_q \subseteq \owF^{\mu_p}_{1,0, p q n}$.
\end{lemma}

In words, Lemma~\ref{le:two_way_ubs} says that any pair of functions in $\calf_q$ are $(p q n)$-close in $\delta_{\mu_p}$-distance, and any function in $\calf_q$ has a one-way zero-error protocol with a single bit of communication. Lemma~\ref{le:two_way_lb} lower bounds the contextually uncertain communication complexity of $\calf_q$ under distribution $\mu_p$.

\begin{lemma}\label{le:two_way_lb}
For any $p=p(n)$, $q=q(n)$ and $\epsilon > 0$, it holds that:
$$\CCU^{\mu_p}_{\frac12 - \epsilon}(\calf_q) \geq 
\gamma \cdot \min\{p\cdot n,(q/2)\cdot n\} -\log(1/\epsilon) + \eta,$$
where $\eta = 2^{-\Theta(q \cdot n)}/\epsilon$, and $\gamma > 0$ is an absolute constant.
\end{lemma}

Note that applying Lemmas \ref{le:two_way_ubs} and \ref{le:two_way_lb} with $\calf = \calf_q$, $\mu = \mu_p$ and $p = q = \sqrt{\delta/n}$  (where $\delta > 0$ is any constant) implies Theorem~\ref{thm:main_lb_non_prod}.

In Subsection~\ref{subsec:napp_two_way_ubs} below, we prove Lemma~\ref{le:two_way_ubs} which follows from two simple propositions. The main part of the rest of this section is dedicated to the proof of Lemma~\ref{le:two_way_lb}. The idea behind the proof of Lemma~\ref{le:two_way_lb} is to reduce the problem of computing $\calf_q$
under $\mu_p$ with contextual uncertainty, into the problem of computing a related function in the standard distributional communication complexity model
(i.e., \emph{without} uncertainty) under a related distribution. We then use the \emph{discrepancy method} to prove a lower bound on the communication
complexity of the new problem. This task itself reduces to
upper bounding the spectral norm of a certain communication matrix. The choice of our underlying distribution then implies a tensor structure for this matrix, which reduces the spectral norm computation to bounding  the largest singular value of an explicit family of $4 \times 4 $ matrices.

We point out that our lower bound in Lemma~\ref{le:two_way_lb} is essentially tight up to a logarithmic factor. Namely, one can show using a simple one-way hashing protocol that for any constant $\epsilon > 0$, $\owCCU^{\mu_p}_{\epsilon}(\calf_q) \le O(r \log{r})$ with $r \triangleq \min\{2p\cdot n, q\cdot n\}$.


\subsection{Proof of Lemma~\ref{le:two_way_ubs}}\label{subsec:napp_two_way_ubs}

Lemma~\ref{le:two_way_ubs} follows from Propositions~\ref{prop:distance} and \ref{prop:owCC} below. We first show that every pair of functions in $\calf_q$ are close under the distribution $\mu_p$.
\begin{proposition}\label{prop:distance}
For every $(f,g) \in \calf_q$, it holds that $\delta_{\mu_p}(f,g) \leq p q n$. 
\end{proposition}
\begin{proof}
Any pair of functions $(f,g) \in \calf_q$ is of the form $f = f_S$ and $g = f_T$ with $|S \triangle T| \le q$. Hence,
\begin{equation*}
\Pr_{(x,y) \sim \mu}[ f(x, y) \neq g(x, y)] = \Pr_{(x,y) \sim \mu}[ f_{S \triangle T}(x \oplus y) =1] \leq 1-(1-p)^{|S \triangle T|} \leq 1-(1-p)^{qn} \leq pqn.
\end{equation*}
\end{proof}

Next, we show that there is a simple one-way communication protocol that allows Alice and Bob to compute $f_S$ (for any $S \subseteq [n]$) with just a single bit of communication.


\begin{proposition}\label{prop:owCC}
$\owCC(f_S) = 1$.
\end{proposition}
\begin{proof}
  Recall that $f_S(x,y) = \oplus_{i \in S} (x_i \oplus y_i)$. We write this as $f_S(x,y) = \left(\oplus_{i \in S} (x_i)\right)
  \oplus \left(\oplus_{i \in S} (y_i)\right)$. This leads to the simple one-way protocol where Alice computes
  $b=\oplus_{i \in S} (x_i)$ and sends the single bit result of the computation to Bob. Bob can now compute $b\oplus
  \left(\oplus_{i \in S} (y_i)\right) = f_S(x,y)$ to obtain the value of $f_S$ (with zero error).
\end{proof}

\subsection{Proof of Lemma~\ref{le:two_way_lb}}
In order to lower bound $\CCU_{\frac12 - \epsilon}^{\mu_p}(\calf_q)$, we define a communication problem in the standard distributional complexity setting that can be reduced to the problem of computing $\calf_q$ under contextual uncertainty. The lower bound in Lemma~\ref{le:two_way_lb} is then obtained by proving a lower bound on the communication complexity of the new problem which is defined as follows:
\begin{itemize}
\item \textbf{Inputs:} Alice's input is a pair $(S,x)$ where $S \subseteq [n]$ and $x \in \zo^n$. Bob's input is a pair $(T,y)$ such that $T \subseteq [n]$ and $y \in \zo^n$. 

\item \textbf{Distribution:} Let $\cald_q$ be a distribution on pairs of Boolean functions $(f,g)$ on $\zo^n \times \zo^n$ defined by the following sampling procedure. To sample $(f,g) \sim \cald_q$, we pick a set $S \subseteq[n]$ uniformly at random and set $f = f_S$. We then pick $T$ to be a $(q/2)$-noisy copy of $S$ and set $g = f_T$. The distribution on the inputs of Alice and Bob is then described by $\nu_{p,q} = \cald_q \times \mu_p$: we sample $(x,y) \sim \mu_p$ and sample $(f,g) \sim \cald_q$. 

\item \textbf{Function:} The goal is to compute the function $F$ given by $F( (S,x), (T,y)) \triangleq f_T(x \oplus y).$
\end{itemize}

Proposition~\ref{prop:reduction_lem} below -- which follows from a simple Chernoff bound -- shows that a protocol computing $\calf_q$ under $\mu_p$ can also be used to compute the function $F$ in the standard distributional model with $((S,x), (T,y)) \sim \nu_{p,q}$, and with the same amount of communication.


\begin{proposition}\label{prop:reduction_lem}
For every $\epsilon > 0$, it holds that $\CCU_{\frac12 - \epsilon}^{\mu_p}(\calf_q) \geq \CC^{\nu_{p,q}}_{\frac 12 - \epsilon + \epsilon'}(F)$ with $\epsilon' = 2^{-\Theta(q \cdot n)}$.
\end{proposition}

In the rest of this section, we will prove the following lower bound on $\CC^{\nu_{p,q}}_{\frac{1}{2}-\epsilon}(F)$, which along with Proposition~\ref{prop:reduction_lem}, implies Lemma~\ref{le:two_way_lb}:
\begin{lemma}\label{le:ucc_lb}
For every $\epsilon > 0$, it holds that
$$\CC^{\nu_{p,q}}_{\frac 12 - \epsilon}(F) \geq \gamma \cdot \min\{p \cdot n,(q/2)\cdot n\} - \log(1/\epsilon),$$
where $\gamma > 0$ is an absolute constant.
\end{lemma}

To prove Lemma~\ref{le:ucc_lb}, without loss of generality, we will set $q = 2p$ and prove a lower bound of $\gamma \cdot p \cdot n$ on the communication complexity\footnote{We can do so because $\CC^{\nu_{p,q}}_{\frac 12 - \epsilon}(F) \geq \CC^{\nu_{r,2r}}_{\frac 12 - \epsilon}(F)$ with $r \triangleq \min(p,q/2)$, which follows from the fact that Alice can always use her private randomness to reduce the correlation between either $(x,y)$ or $(S,T)$.}. So henceforth, we denote $\nu_{p} \triangleq \nu_{p,2p}$. The proof will use the \emph{discrepancy bound} which is a well-known method for proving lower bounds on distributional communciation complexity in the standard model.

\begin{definition}[Discrepancy; \cite{KN97}]
Let $F$ and $\nu_p$ be as above and let $R$ be any rectangle (i.e., any set of the form $R = C \times D$ where $C,D \subseteq \{0,1\}^{2n}$). Denote
\begin{align*}
\Disc_{\nu_p}(R,F) \triangleq
 \bigg| &\Pr_{\nu_p}\big[F((S,x),(T,y)) = 0, ((S,x),(T,y)) \in R\big] - \\
 & \Pr_{\nu_p}\big[F((S,x),(T,y)) = 1, ((S,x),(T,y)) \in R\big] \bigg|.
\end{align*}
The discrepancy of $F$ according to $\nu_p$ is $\Disc_{\nu_p}(F) \triangleq \max_{R} \Disc_{\nu_p}(R,F)$ where the maximum is over all rectangles $R$.
\end{definition}

The next known proposition relates distributional communication complexity to discrepancy.

\begin{proposition}[\cite{KN97}]\label{prop:disc_bd}
For any $\epsilon > 0$, it holds that $ \CC^{\nu_p}_{\frac 12 - \epsilon}(F) \geq \log(2\epsilon / \Disc_{\nu_p}(F))$.
\end{proposition}

We will prove the following lemma.

\begin{lemma}\label{le:disc_up_bd}
$\Disc_{\nu_p}(F) \le 2^{-\gamma \cdot p \cdot n}$ for some absolute constant $\gamma > 0$.
\end{lemma}

Note that Lemma~\ref{le:disc_up_bd} and Proposition~\ref{prop:disc_bd} put together immediately imply Lemma~\ref{le:ucc_lb}. The proof of Lemma~\ref{le:disc_up_bd} uses some standard facts about the spectral properties of matrices and their tensor powers that we next recall. Let $A \in \mathbb{R}^{d \times d}$ be a real square matrix. Then, $v \in \mathbb{R}^d$ is said to be an eigenvector of $A$ with eigenvalue $\lambda \in \mathbb{R}$ if $Av = \lambda v$. If $A$ is furthermore (symmetric) positive semi-definite, then all its eigenvalues are real and non-negative. We can now define the spectral norm of a (not necessarily symmetric) matrix.

\begin{definition}
The spectral norm of a matrix $A \in \mathbb{R}^{d \times d}$ is given by $\| A \| \triangleq \sqrt{\lambda_{max}(A^T A)}$, where $\lambda_{max}(A^T A)$ is the largest eigenvalue of $A^T A$.
\end{definition}

Also, recall that given a matrix $A \in \mathbb{R}^{d \times d}$ and a positive integer $t$, the tensor power matrix $A^{\otimes t} \in \mathbb{R}^{d^t \times d^t}$ is defined by $(A^{\otimes t})_{(i_1,\dots,i_t)} = \prod_{j = 1}^t A_{i_j}$ for every $(i_1,\dots,i_t) \in \{1,\dots,d\}^t$. We will use the following standard fact which in particular says that the spectral norm is multiplicative with respect to tensoring.

\begin{fact}\label{fa:spec_prop}
For any matrix $A \in \mathbb{R}^{d \times d}$, vector $u \in \mathbb{R}^d$, scalar $c \in \mathbb{R}$ and positive integer $t$, we have
\begin{enumerate}
\item $\| c A \| = |c| \cdot \|A \|$.
\item $\| A^{\otimes t} \| = \| A \|^t$.
\item $\| A u \|_2 \le \| A \| \cdot \| u \|_2$, where for any vector $w \in \mathbb{R}^d$,  $\| w \|_2$ denotes the \emph{Euclidean norm} of $w$, i.e., $\| w \|_2 \triangleq \sqrt{ \sum_{i=1}^d w_i^2}$. 
\end{enumerate}
\end{fact}

The next lemma upper bounds the spectral norm of an explicit family of $4 \times 4$ matrices that will be used in the proof of Lemma~\ref{le:disc_up_bd}. Looking ahead, it is crucial for our purposes that the coefficient of $a$ in the right-hand side of \Cref{eq:full_N_bd} is a constant strictly smaller than $2$.
\begin{lemma}\label{le:sp_norm_bd}
Let $a \in (0,1)$ be a real number and $N \triangleq N(a) \triangleq \left[ \begin{array}{cccc}1 & a & a & -a^2 \\ 
a & 1 & -a^2 & a \\ 
a & a^2 & 1 & -a \\ 
a^2 & a & -a & 1 \end{array} \right]$. Then,
\begin{equation}\label{eq:full_N_bd}
\| N \|_2 \le 1 + \sqrt{2} \cdot a + a^2 + \frac{a^4}{2} + \frac{a^5}{\sqrt{2}}.
\end{equation}
\end{lemma}

\begin{proof}
One can verify that
$$ N^{T} N = \left[ \begin{array}{cccc} (a^2+1)^2 & 2a(a^2+1) & 2a(1-a^2) & 0 \\ 
2a(a^2+1) & (a^2+1)^2 & 0 & 2a(1-a^2) \\ 
2a(1-a^2) & 0 & (a^2+1)^2 & -2a(a^2+1) \\ 
0 & 2a(1-a^2) & -2a(a^2+1) & (a^2+1)^2 \end{array} \right]. $$
Assuming that $a \in (0,1)$, one can also verify that $N^{T} N$ has as eigenvectors
$$ v_1 \triangleq  \left[ \begin{array}{c}\frac{\sqrt{2(a^4+1)}}{1-a^2} \\ 
\frac{a^2+1}{1-a^2} \\ 
1 \\ 
0 \end{array} \right], ~ v_2 \triangleq \left[ \begin{array}{c} \frac{a^2+1}{1-a^2} \\ 
\frac{\sqrt{2(a^4+1)}}{1-a^2} \\ 
0 \\ 
1 \end{array} \right] \text{ with eigenvalue } \lambda_1(a) \triangleq  2a^2+a^4+2a\sqrt{2(a^4+1)}+1,$$
$$ \text{and } v_3 \triangleq  \left[ \begin{array}{c} \frac{\sqrt{2(a^4+1)}}{a^2-1} \\ 
\frac{a^2+1}{1-a^2} \\ 
1 \\ 
0 \end{array} \right], ~ v_4 \triangleq \left[ \begin{array}{c} \frac{a^2+1}{1-a^2} \\ 
\frac{\sqrt{2(a^4+1)}}{a^2-1} \\ 
0 \\ 
1 \end{array} \right] \text{ with eigenvalue } \lambda_2(a) : = 2a^2+a^4-2a\sqrt{2(a^4+1)}+1.$$

Note that for any value of $a \in (0,1)$, the vectors $v_1$, $v_2$, $v_3$ and $v_4$ are linearly independent and each of the eigenvalues $\lambda_1(a)$ and $\lambda_2(a)$ has multiplicity $2$. Moreover, we have that $\lambda_1(a) \geq \lambda_2(a)$. Hence,
$$\| N \| = \sqrt{\lambda_1(a)} = \sqrt{2a^2+a^4+2a\sqrt{2(a^4+1)}+1}.$$

Applying twice the fact that $\sqrt{1+x} \le 1+x/2$ for any $x \geq -1$, we get that
\begin{align*}
\| N \| &= \sqrt{1+2a^2+a^4+2a\sqrt{2}\sqrt{1+a^4}}\\ 
&\le 1+a^2+\frac{a^4}{2}+a\sqrt{2}\sqrt{1+a^4}\\ 
&\le 1+a^2+\frac{a^4}{2}+a\sqrt{2}(1+\frac{a^4}{2})\\ 
&= 1 + a\sqrt{2} + a^2 + \frac{a^4}{2} + \frac{a^5}{\sqrt{2}}.
\end{align*}

\end{proof}

We are now ready to prove Lemma~\ref{le:disc_up_bd}.

\begin{proofof}{{\bf Lemma~\ref{le:disc_up_bd}}}
Fix any rectangle $R = C \times D$ where $C,D \subseteq \{0,1\}^{2n}$. We wish to show that $\Disc_{\nu_p}(R,F) \le 2^{-\gamma \cdot p \cdot n}$. First, note that $\Disc_{\nu_p}(R,F) = |1_C M 1_D|$ where $1_C$ and $1_D$ are the $0/1$ indicator vectors of $C$ and $D$ (respectively) and $M$ is the $2^{2n} \times 2^{2n}$ real matrix defined by\footnote{We here use the symbols $S$ and $T$ to denote both subsets of $[n]$ and the corresponding $0/1$ indicator vectors.}
\begin{align*}
M_{((S,x),(T,y))} &\triangleq \nu_p((S,T),(x,y)) \cdot (-1)^{f_T(x \oplus y)}\\ 
&= \frac{1}{2^{2n}} (1-p)^{2n} (-1)^{\langle T , x \oplus y \rangle} (\frac{p}{1-p})^{|S \oplus T| + |x \oplus y|}
\end{align*}
for every $S,x,T,y \in \{0,1\}^n$. Letting $a \triangleq p/(1-p)$, we can write
\begin{align*}
M_{((S,x),(T,y))} = \frac{1}{2^{2n}} (1-p)^{2n} (N^{\otimes n})_{((S,x),(T,y))}
\end{align*}
with $N = N(a)$ being the $4 \times 4$ real matrix defined by\footnote{In \Cref{eq:N_def}, $T_1(x_1 \oplus y_1)$ denotes the product of the bit $T_1$ and the bit $(x_1 \oplus y_1)$. Moreover, since $(S_1 \oplus T_1)$ is a single bit, its Hamming weight $|S_1 \oplus T_1|$ is the same as its bit-value, and similarly for $(x_1 \oplus y_1)$.}
\begin{equation}\label{eq:N_def}
N_{((S_1,x_1),(T_1,y_1))} \triangleq (-1)^{T_1(x_1 \oplus y_1)} a^{|S_1 \oplus T_1| + |x_1 \oplus y_1|}
\end{equation}
for all $S_1,x_1,T_1,y_1 \in \{0,1\}$. Using the third property listed in Fact~\ref{fa:spec_prop}, we get
\begin{align}\label{eq:disc_M}
\Disc_{\nu_p}(R,F) = |1_C M 1_D| \le \| 1_C \|_2 \cdot \| M \| \cdot \| 1_D \|_2 \le \sqrt{2^{2n}} \cdot \| M \| \cdot \sqrt{2^{2n}} = 2^{2n} \cdot \| M \|
\end{align}

We now use the first two properties listed in Fact~\ref{fa:spec_prop} to relate $\| M \|$ to $\| N \|$ as follows:
\begin{equation}\label{eq:M_N_bd}
\| M \| = \| \frac{1}{2^{2n}} (1-p)^{2n} N^{\otimes n} \| = \frac{1}{2^{2n}} (1-p)^{2n} \| N \|^n.
\end{equation}

Using \Cref{eq:N_def}, we can check that

$$N = N(a) = \left[ \begin{array}{cccc}1 & a & a & -a^2 \\ 
a & 1 & -a^2 & a \\ 
a & a^2 & 1 & -a \\ 
a^2 & a & -a & 1 \end{array} \right].$$

Applying Lemma~\ref{le:sp_norm_bd} with $a= p/(1-p)$ and $p$ sufficiently small (e.g., less than $1/10$), we get
\begin{equation}\label{eq:N_up_bd}
\| N \| \le 1 + \sqrt{2} \cdot (\frac{p}{1-p}) + O(p^2).
\end{equation}

Combining \Cref{eq:disc_M,eq:M_N_bd,eq:N_up_bd} above, we conclude that
\begin{align*}
\Disc_{\nu_p}(R,F) &\le (1-p)^{2n} \cdot \big(1 + \sqrt{2} \cdot (\frac{p}{1-p}) + O(p^2)\big)^n\\ 
&= \bigg[ (1-p) \cdot \big(1 + p \cdot (\sqrt{2}-1)+O(p^2) \big) \bigg]^n\\ 
&= \bigg[1 - p \cdot (2-\sqrt{2})+O(p^2) \bigg]^n\\ 
&\le 2^{-\gamma \cdot p \cdot n}
\end{align*}
for some absolute constant $\gamma > 0$.
\end{proofof}

%% file: conc.tex
\section{Conclusion and Future Directions}\label{sec:conc}

In this work, we introduced and studied a simple model illustrating the role of context in communication and the challenge posed by uncertainty of knowledge of context.

On the technical side, it would be interesting to determine the correct exponent of $I(X;Y)$ in Theorem~\ref{thm:one-way}. Theorems~\ref{thm:one-way} and \ref{thm:main_lb_non_prod} imply that this exponent is between $1/2$ and $1$.

It would also be interesting to prove an analogue of Theorem~\ref{thm:one-way} for two-way protocols. Our proof of Theorem~\ref{thm:one-way} uses in particular the fact that any low-communication one-way protocol in the standard distributional
communication model should have a certain canonical form: to compute $g(x,y)$, Alice tries to describe the entire function $g(x,\cdot)$ to Bob,
and this does not create a huge overhead in communication. Coming up with a canonical form of two-way protocols that somehow changes gradually as we morph from $g$ to $f$ seems to be the essence of the challenge in extending Theorem~\ref{thm:one-way} to the two-way setting.

On the more conceptual side, arguably, the model considered in this work is a fairly realistic one: communication has some goals in mind
which we model by letting Bob be interested in a specific function of the joint information that Alice and
Bob possess. Moreover, it is a fairly natural model to posit that the two are not in perfect synchronization
about the function that Bob is interested in, but Alice can estimate the function in some sense.
One aspect of our model that can be further refined is the specific notion of distance that quantifies the gap between Bob's function and Alice's estimate. In this work, we chose the Hamming distance which forms a good first starting point. We believe that it is interesting to propose
and study other models of distance between functions that more accurately capture natural forms of uncertainty.

Finally, we wish to emphasize the mix of adversarial and probabilistic elements in our uncertainty model --- the adversary picks $(f,g)$ whereas the inputs $(x,y)$ are picked from a distribution. We believe that richer
mixtures of adversarial and probabilistic elements could lead to broader settings of modeling and coping
with uncertainty --- the probabilistic elements offer efficient possibilities that are often immediately ruled out
by adversarial choices, whereas the adversarial elements prevent the probabilistic assumptions from being
too precise.